\documentclass[10pt]{article}
\usepackage[letterpaper,hmargin=1in,vmargin=1.25in]{geometry}
\newcommand{\keywords}{\textbf{Keywords.}~}
\usepackage{ifpdf}
\usepackage[utf8]{inputenc}
\usepackage[english]{babel}
\usepackage{amsmath,amsfonts,amsthm,mathtools,epstopdf,stmaryrd,nicefrac}
\usepackage{color,xspace,xcolor,graphicx,paralist,wrapfig,balance}
\usepackage{relsize,alltt,multirow,booktabs,paralist,fancyvrb}
\ifpdf
\usepackage[pdftex,
            pdfauthor={Pablo Rauzy, Sylvain Guilley},
						colorlinks=false,pdfborder={0 0 0},
            pdftitle={Formal Analysis of CRT-RSA Vigilant's Countermeasure Against the BellCoRe Attack}
            ]{hyperref}
\usepackage{xmpincl}
\includexmp{CC_Attribution_4.0_International}
\else
\usepackage{hyperref}
\fi

\usepackage[ruled,linesnumbered]{algorithm2e}
\newcommand{\algorithmConfig}{
  \DontPrintSemicolon
	\SetKw{Or}{or}
	\SetKwInOut{Input}{Input}
	\SetKwInOut{Output}{Output}
}
\makeatletter
\newcommand{\removelatexerror}{\let\@latex@error\@gobble}
\makeatother

\theoremstyle{definition}
\newtheorem{definition}{Definition}
\newtheorem{remark}{Remark}
\theoremstyle{theorem}
\newtheorem{proposition}{Proposition}

\newcommand{\finja}{\textsf{finja}\xspace}

\newcommand{\Z}{\mathbb{Z}}
\newcommand{\ie}{\textit{i.e.}}
\newcommand{\eg}{\textit{e.g.}}
\newcommand{\etal}{\textit{et al.}\xspace}

\title{\textbf{Formal Analysis of CRT-RSA Vigilant's Countermeasure Against the BellCoRe Attack}\\
  {\smaller \smaller A Pledge for Formal Methods in the Field of Implementation Security}}
\author{Pablo Rauzy --- Sylvain Guilley\\
Institut Mines-Télécom ; Télécom ParisTech ; CNRS LTCI\\
\{{\it firstname}.{\it lastname}\}@telecom-paristech.fr}
\date{}
\begin{document}

\maketitle

\begin{abstract}
In our paper at PROOFS~2013, we formally studied a few known countermeasures to protect CRT-RSA against the BellCoRe fault injection attack.
However, we left Vigilant's countermeasure and its alleged repaired version by Coron \etal as future work, because the arithmetical framework of our tool was not sufficiently powerful.
In this paper we bridge this gap and then use the same methodology to formally study both versions of the countermeasure.
We obtain surprising results, which we believe demonstrate the importance of formal analysis in the field of implementation security.
Indeed, the original version of Vigilant's countermeasure is actually broken, but not as much as Coron \etal thought it was.
As a consequence, the repaired version they proposed can be simplified.
It can actually be simplified even further as two of the nine modular verifications happen to be unnecessary.
Fortunately, we could formally prove the simplified repaired version to be resistant to the BellCoRe attack, which was considered a ``challenging issue'' by the authors of the countermeasure themselves.
\end{abstract}

\keywords
RSA (\textit{Rivest, Shamir, Adleman}),
CRT (\textit{Chinese Remainder Theorem}),
fault injection,
BellCoRe (\textit{Bell Communications Research}) attack,
formal proof,
OCaml.

\section{Introduction}

Private information protection is a highly demanded feature, especially in the context of global defiance against most infrastructures, assumed to be controlled by governmental agencies.
Properly used cryptography is known to be a key building block for secure information exchange.
However, in addition to the threat of cyber-attacks, implementation-level hacks are also to be considered seriously.
This article deals specifically with the protection of a \emph{decryption} or \emph{signature} crypto-system (called RSA~\cite{DBLP:journals/cacm/RivestSA78}) in the presence of hardware attacks (\eg, we assume the attacker can alter the RSA computation while it is being executed).

It is known since 1997~\cite{boneh-fault} that injecting faults during the computation of CRT-RSA could yield to malformed signatures that expose the prime factors ($p$ and $q$) of the public modulus ($N=p \cdot q$).
Notwithstanding, computing without the fourfold acceleration conveyed by the CRT is definitely not an option in practical applications.
Therefore, many countermeasures have appeared that consist in step-wise internal checks during the CRT computation.
Last year we formally studied some of them~\cite{cryptoeprint:2013:506}.
We were able to formally prove that the unprotected implementation of CRT-RSA as well as the countermeasure proposed by Shamir~\cite{shamir-patent-rsa-crt} are broken, and that Aumüller \etal's countermeasure~\cite{DBLP:conf/ches/AumullerBFHS02} is resistant to the BellCoRe attack.
However, we were not able to study Vigilant's countermeasure~\cite{patent-vigilant_crtrsa,DBLP:conf/ches/Vigilant08} or its repaired version by Coron \etal~\cite{DBLP:conf/fdtc/CoronGMPV10}, because the tool we developed lacked the ability to work with arithmetical properties necessary to handle these countermeasures.
In particular, our framework was unaware of the easiness to compute the discrete logarithm in some composite degree residuosity classes. % Cf. \eqref{paillier}
These difficulties were foreseen by Coron \etal themselves in the conclusion of their paper~\cite{DBLP:conf/fdtc/CoronGMPV10}:
\begin{quote}
\emph{``Formal proof of the FA-resistance of Vigilant's scheme including our countermeasures is still an open (and challenging) issue.''}
\end{quote}

This is precisely the purpose of this paper.
The state-of-the-art of formal proofs of Vigilant's countermeasure is the work of Christofi, Chetali, Goubin, and Vigilant~\cite{JCEN-Christofi13}.
However, for tractability reasons, the proof is conducted on reduced versions of the algorithms parameters.
One fault model is chosen authoritatively (\emph{viz}. the zeroization of a complete intermediate data), which is a strong assumption.
In addition, the verification is conducted on a specific implementation in pseudocode, hence concerns about its portability after compilation into machine-level code.

\paragraph{Contributions.}
We improved \finja (our tool based on symbolic computation in the framework of modular arithmetic~\cite{cryptoeprint:2013:506}) to enable the formal study of CRT-RSA Vigilant's countermeasure against the BellCoRe attack.
The \finja tool allows a full fault coverage of CRT-RSA algorithm, thereby keeping the proof valid even if the code is transformed (\eg, optimized, compiled, partitioned in software/hardware, or equipped with dedicated countermeasures).
We show that the original countermeasure~\cite{DBLP:conf/ches/Vigilant08} is indeed broken, but not as much as Coron \etal thought it was when they proposed a \emph{manually} repaired version of it~\cite{DBLP:conf/fdtc/CoronGMPV10}.
We simplify the repaired version accordingly, and formally prove it resistant to the BellCoRe attack under a fault model that considers \emph{permanent faults} (in memory) and \emph{transient faults} (one-time faults, even on copies of the secret key parts, \eg, during their transit on buses or when they reside on register banks), with or without forcing at zero, and with possibly faults at various locations.
Thanks to the formal analysis, we are able to simplify the countermeasure even further: two of the nine checks are unnecessary.

\paragraph{Organization of the paper.}
We recall CRT-RSA cryptosystem and the BellCoRe attack in Sec.~\ref{sec-rsa}.
Vigilant's countermeasure and its variant by Coron \etal are exposed in Sec.~\ref{sec-vigilant}.
In Sec.~\ref{sec-methods} we detail the modular arithmetic framework used by our tool.
The results of our analysis are presented in Sec.~\ref{sec-analysis}.
Finally, conclusions and perspectives are given in Sec.~\ref{sec-concl}.
After that, Appx.~\ref{app-practice} details the practical issues related to fault injection analysis, for single- and multiple-fault attacks.

\section{CRT-RSA and the BellCoRe Attack}
\label{sec-rsa}

This section recaps known results about fault attacks on CRT-RSA (see also \cite{koc_rsa}, \cite[Chap. 3]{Intro_HOST} and \cite[Chap. 7 \& 8]{fabook}).
Its purpose is to settle the notions and the associated notations that will be used in the later sections (that contain novel contributions).

\subsection{RSA}

RSA is both an \emph{encryption} and a \emph{signature} scheme.
It relies on the identity that for all message $0\leq M<N$,
$(M^d)^e \equiv M \mod N$, where $d \equiv e^{-1} \mod \varphi(N)$, by Euler's theorem%
\footnote{We use the usual convention in all mathematical equations, namely that the ``$\text{mod}$'' operator has the lowest binding precedence,
\ie, $a \times b \mod c \times d$ represents the element $a \times b$ in $\Z_{c \times d}$.}.
In this equation, $\varphi$ is Euler's totient function, equal to $\varphi(N)=(p-1) \cdot (q-1)$ when $N=p \cdot q$ is a composite number, product of two primes $p$ and $q$.
For example, if Alice generates the signature $S=M^d \mod N$,
then Bob can verify it by computing $S^e \mod N$, which must be equal to $M$ unless Alice is only pretending to know $d$.
Therefore $(N,d)$ is called the private key, and $(N,e)$ the public key.
In this paper, we are not concerned about the key generation step of RSA,
and simply assume that $d$ is an unknown number in $\llbracket 1, \varphi(N)=(p-1) \cdot (q-1)\llbracket$.
Actually, $d$ can also be chosen equal to the smallest value $e^{-1} \mod \lambda(N)$,
where $\lambda(N) = \frac{(p-1) \cdot (q-1)}{\gcd(p-1, q-1)}$ is the Carmichael function (see PKCS \#1 v2.1, \S 3.1).
% ftp://ftp.rsasecurity.com/pub/pkcs/pkcs-1/pkcs-1v2-1.pdf

\subsection{CRT-RSA}

The computation of $M^d \mod N$ can be speeded-up by a factor four by using the Chinese Remainder Theorem (CRT).
Indeed, figures modulo $p$ and $q$ are twice as short as those modulo $N$.
For example, for $2,048$\,bits RSA, $p$ and $q$ are $1,024$\,bits long.
The CRT-RSA consists in computing $S_p = M^d \mod p$ and $S_q = M^d \mod q$,
which can be recombined into $S$ with a limited overhead.
Due to the little Fermat theorem (special case of the Euler theorem when the modulus is a prime),
$S_p = (M \mod p)^{d \mod (p-1)} \mod p$.
This means that in the computation of $S_p$, the processed data have $1,024$\,bits,
and the exponent itself has $1,024$\,bits (instead of $2,048$\,bits).
Thus the multiplication is four times faster and the exponentiation eight times faster.
However, as there are two such exponentiations (modulo $p$ and $q$), the overall CRT-RSA is roughly speaking four times faster than RSA computed modulo $N$.

This acceleration justifies that CRT-RSA is always used if the factorization of $N$ as $p \cdot q$ is known.
In CRT-RSA, the private key is a more rich structure than simply $(N,d)$:
it is actually comprised of the $5$-tuple $(p, q, d_p, d_q, i_q)$, where:
\begin{itemize}
\item $d_p \doteq d \mod (p-1)$,
\item $d_q \doteq d \mod (q-1)$, and
\item $i_q \doteq q^{-1} \mod p$.
\end{itemize}
The CRT-RSA algorithm is presented in Alg.~\ref{alg-crt-rsa-naive}.
It is straightforward to check that the signature computed at line~\ref{alg-crt-rsa-naive-recombination}
belongs to $\llbracket 0, p \cdot q -1 \rrbracket$.
Consequently, no reduction modulo $N$ is necessary before returning $S$.

\begin{algorithm}
\algorithmConfig
\Input{Message $M$, key $(p, q, d_p, d_q, i_q)$}
\Output{Signature $M^d \mod N$}
\BlankLine
$S_p = M^{d_p} \mod p$ \tcp*[r]{Signature modulo $p$}
$S_q = M^{d_q} \mod q$ \tcp*[r]{Signature modulo $q$}
$S = S_q + q \cdot (i_q \cdot (S_p-S_q) \mod p)$ \tcp*[r]{Recombination} \label{alg-crt-rsa-naive-recombination}
\Return $S$ \;
\caption{Unprotected CRT-RSA}
\label{alg-crt-rsa-naive}
\end{algorithm}

\subsection{The BellCoRe Attack}

In 1997, an dreadful remark has been made by Boneh, DeMillo and Lipton~\cite{boneh-fault}, three staff of BellCoRe:
Alg.~\ref{alg-crt-rsa-naive} could reveal the secret primes $p$ and $q$ if the line 1 or 2 of the computation is faulted, even in a very random way.
The attack can be expressed as the following proposition.
\begin{proposition}[Original BellCoRe attack]
If the intermediate variable $S_p$ (resp. $S_q$) is returned faulted as $\widehat{S_p}$ (resp. $\widehat{S_q}$)\footnote{In other papers related to faults, the faulted variables (such as $X$) are noted either with a star ($X^*$) or a tilde ($\tilde{X}$); in this paper, we use a hat, as it can stretch, hence cover the adequate portion of the variable. For instance, it allows to make an unambiguous difference between a faulted data raised at some power and a fault on a data raised at a given power (contrast $\widehat{X}^e$ with $\widehat{X^e}$).},
then the attacker gets an erroneous signature $\widehat{S}$,
and is able to recover $p$ (resp. $q$) as $\gcd(N, S-\widehat{S})$.
\label{pro-bellcore2faults}
\end{proposition}
\begin{proof}
For any integer $x$, $\gcd(N, x)$ can only take $4$ values:
\begin{itemize}
\item $1$, if $N$ and $x$ are coprime,
\item $p$, if $x$ is a multiple of $p$,
\item $q$, if $x$ is a multiple of $q$,
\item $N$, if $x$ is a multiple of both $p$ and $q$, \ie, of $N$.
\end{itemize}
In Alg.~\ref{alg-crt-rsa-naive}, if $S_p$ is faulted (\ie, replaced by $\widehat{S_p} \neq S_p$), then\\
$S-\widehat{S} = q \cdot (( i_q \cdot (S_p-S_q)\mod p) - ( i_q \cdot (\widehat{S_p}-S_q) \mod p))$,
and thus $\gcd(N, S-\widehat{S}) = q$.\\
If $S_q$ is faulted (\ie, replaced by $\widehat{S_q} \neq S_q$), then\\
$S-\widehat{S} \equiv (S_q - \widehat{S_q}) - ( q \mod p) \cdot i_q \cdot (S_q - \widehat{S_q}) \equiv 0 \mod p$ because $(q \mod p) \cdot i_q \equiv 1 \mod p$, and thus $S-\widehat{S}$ is a multiple of $p$.
Additionally, $S-\widehat{S}$ is not a multiple of $q$.\\
So, $\gcd(N, S-\widehat{S})=p$.
\end{proof}

\section{Vigilant Countermeasure Against the BellCoRe Attack}
\label{sec-vigilant}

Fault attacks on RSA can be thwarted simply by refraining from implementing the CRT.
If this is not affordable, then the signature can be verified before being outputted.
Such protection is efficient in practice, but is criticized for two reasons.
First of all, it requires an access to $e$;
second, the performances are incurred by the extra exponentiation needed for the verification.
This explains why several other countermeasures have been proposed.

Until recently, none of these countermeasures had been proved.
In 2013, Christofi, Chetali, Goubin, and Vigilant~\cite{JCEN-Christofi13} formally proved an implementation of Vigilant's countermeasure.
However, for tractability purposes, the proof is conducted on reduced versions of the algorithms parameters.
One fault model is chosen authoritatively (the zeroization of a complete intermediate data), which is a strong assumption.
In addition, the verification is conducted on a specific implementation in pseudocode, hence concerns about its portability after compilation into machine-level code.
The same year, we formally studied~\cite{cryptoeprint:2013:506} the ones of Shamir~\cite{shamir-patent-rsa-crt} and Aumüller \etal~\cite{DBLP:conf/ches/AumullerBFHS02} using a tool based on the framework of modular arithmetic, thus offering a full fault coverage on CRT-RSA algorithm, thereby keeping the proof valid even if the code is transformed (\eg, optimized, compiled or partitioned in software/hardware).
We proved the former to be broken and the latter to be resistant to the BellCoRe attack.

\subsection{Vigilant's Original Countermeasure}
\label{sec-vigilant-orig}

In his paper at CHES 2008, Vigilant~\cite{DBLP:conf/ches/Vigilant08} proposed a new method to protect CRT-RSA.
Its principle is to compute $M^d \mod N$ in $\Z_{Nr^2}$ where $r$ is a random integer that is coprime with $N$.
Then $M$ is transformed into $M^*$ such that
\[ M^* \equiv
\begin{cases}
  M \mod N,\\
  1+r \mod r^2,
\end{cases}
\]
which implies that
\[ S^* = M^{*d} \mod Nr^2 \equiv
\begin{cases}
  M^d \mod N,\\
  1+dr \mod r^2.
\end{cases}
\]
The latter results are based on the binomial theorem, which states that:
\begin{align}
(1+r)^d
&= \sum_{k=0}^d {d \choose k}r^k \notag\\
&= 1+dr + {d \choose 2} r^2 + \text{higher powers of $r$},
\end{align}
which simplifies to $1+dr$ in the ring $\Z_{r^2}$.
This property can be seen as a special case of an easy computation of a discrete logarithm in a composite degree residuosity class
(refer for instance to~\cite{DBLP:conf/eurocrypt/Paillier99}).
% = \eqref{paillier} % http://en.wikipedia.org/wiki/Paillier_cryptosystem
Therefore the result $S^*$ can be checked for consistency modulo $r^2$.
If the verification $S^* \stackrel{{\smaller ?}}{=} 1+dr \mod r^2$ succeeds, then the final result $S = S^* \mod N$ is returned.
Concerning the random numbers used in the countermeasure, its author recommend using a $32$\,bits integer for $r$, and $64$\,bits for the $R_i$s.
The original algorithm for Vigilant's countermeasure is presented in Alg.~\ref{alg-crt-rsa-vigilant}.

% Notice that (p-1) can be computed as a p & (~1) since p is odd
\begin{algorithm}
\algorithmConfig
{\smaller % ePrint
\Input{Message $M$, key $(p, q, d_p, d_q, i_q)$}
\Output{Signature $M^d \mod N$, or $\mathsf{error}$}
\BlankLine
Choose random numbers $r$, $R_1$, $R_2$, $R_3$, and $R_4$. \;
\BlankLine
$p' = p r^2$ \;
$M_p = M \mod p'$ \;
$i_{pr} = p^{-1} \mod r^2$ \;
$B_p = p \cdot i_{pr}$ \;
$A_p = 1 - B_p \mod p'$ \;
$M'_p = A_p M_p + B_p \cdot (1 + r) \mod p'$ \;
\BlankLine
\If{$M'_p \not\equiv M \mod p$}{
  \Return$\mathsf{error}$
}
\BlankLine
$d'_p = d_p + R_1 \cdot (p - 1)$ \label{weak-p1} \;
$S_{pr} = {M'_p}^{d'_p} \mod p'$ \;
\BlankLine
\If{$d'_p \not\equiv d_p \mod p - 1$}{
  \Return$\mathsf{error}$
} \label{weak-p2}
\BlankLine
\If{$B_p S_{pr} \not\equiv B_p \cdot (1 + d'_p r) \mod p'$}{
  \Return$\mathsf{error}$
}
\BlankLine
$S'_p = S_{pr} - B_p \cdot (1 + d'_p r - R_3)$ \;
\BlankLine
$q' = q r^2$ \;
$M_q = M \mod q'$ \;
$i_{qr} = q^{-1} \mod r^2$ \;
$B_q = q \cdot i_{qr}$ \;
$A_q = 1 - B_q \mod q'$ \;
$M'_q = A_q M_q + B_q \cdot (1 + r) \mod q'$ \;
\BlankLine
\If{$M'_q \not\equiv M \mod q$}{
  \Return$\mathsf{error}$
}
\BlankLine
\If{$M_p \not\equiv M_q \mod r^2$}{
  \Return$\mathsf{error}$
}
\BlankLine
$d'_q = d_q + R_2 \cdot (q - 1)$ \label{weak-q1} \;
$S_{qr} = {M'_q}^{d'_q} \mod q'$ \;
\BlankLine
\If{$d'_q \not\equiv d_q \mod q - 1$}{
  \Return$\mathsf{error}$
} \label{weak-q2}
\BlankLine
\If{$B_q S_{qr} \not\equiv B_q \cdot (1 + d'_q r) \mod q'$}{
  \Return$\mathsf{error}$
}
\BlankLine
$S'_q = S_{qr} - B_q \cdot (1 + d'_q r - R_4)$ \;
\BlankLine
$S = S'_q + q \cdot (i_q \cdot (S'_p - S'_q) \mod p')$ \;
$N = p q$ \label{weak-pq1} \;
\BlankLine
\If{$N \cdot (S - R_4 - q \cdot i_q \cdot (R_3 - R_4)) \not\equiv 0 \mod N r^2$}{
  \Return$\mathsf{error}$
} \label{weak-pq2}
\BlankLine
\If{$q \cdot i_q \not\equiv 1 \mod p$}{
  \Return$\mathsf{error}$
}
\BlankLine
\Return$S \mod N$ \;
\caption{Vigilant's CRT-RSA}
} % ePrint
\label{alg-crt-rsa-vigilant}
\end{algorithm}

\subsection{Coron \etal Repaired Version}
\label{sec-vigilant-coron}

At FDTC 2010, Coron \etal~\cite{DBLP:conf/fdtc/CoronGMPV10} proposed some corrections to Vigilant's original countermeasure.
They claimed to have found three weaknesses in integrity verifications:
\begin{enumerate}
\item \label{weak-p} one for the computation modulo $p-1$;
\item \label{weak-q} one for the computation modulo $q-1$;
\item \label{weak-pq} and one in the final check modulo $Nr^2$.
\end{enumerate}

They propose a fix for each of these weaknesses.
To correct weakness \#\ref{weak-p}, lines \ref{weak-p1} to \ref{weak-p2} in Alg.~\ref{alg-crt-rsa-vigilant} becomes:

{\removelatexerror
\begin{algorithm}[H]
\algorithmConfig
$p_{minusone} = p - 1$ \label{pminusone_1} \;
$d'_p = d_p + R_1 \cdot p_{minusone}$ \;
$S_{pr} = {M'_p}^{d'_p} \mod p'$ \;
$p_{minusone} = p - 1$ \label{pminusone_2} \;
\If{$d'_p \not\equiv d_p \mod p_{minusone}$}{
  \Return$\mathsf{error}$
}
\caption{Vigilant's CRT-RSA: Coron \etal fix \#\ref{weak-p}}
\label{alg-crt-rsa-vigilant-fix-p}
\end{algorithm}
}
\medskip

Similarly, for weakness \#\ref{weak-q}, lines \ref{weak-q1} to \ref{weak-q2} in Alg.~\ref{alg-crt-rsa-vigilant}:

{\removelatexerror
\begin{algorithm}[H]
\algorithmConfig
$q_{minusone} = q - 1$ \;
$d'_q = d_q + R_2 \cdot q_{minusone}$ \;
$S_{qr} = {M'_q}^{d'_q} \mod q'$ \;
$q_{minusone} = q - 1$ \;
\If{$d'_q \not\equiv d_q \mod q_{minusone}$}{
  \Return$\mathsf{error}$
}
\caption{Vigilant's CRT-RSA: Coron \etal fix \#\ref{weak-q}}
\label{alg-crt-rsa-vigilant-fix-q}
\end{algorithm}
}
\smallskip
\noindent
These two fixes immediately looked suspicious to us.
Indeed, Coron \etal suppose that the computations of $p-1$ and $q-1$ are factored as an optimization, while they are not in the original version of Vigilant's countermeasure.
This kind of optimization, called ``common subexpression elimination'', is very classical but is, as for most speed-oriented optimization, not a good idea in security code.
We will see in Sec.~\ref{sec-analysis} that the original algorithm, which did not include this optimization, was not at fault here.
\medskip

The third fix proposed by Coron \etal consists in transforming lines \ref{weak-pq1} to \ref{weak-pq2} of Alg.~\ref{alg-crt-rsa-vigilant} into:

{\removelatexerror
\begin{algorithm}[H]
\algorithmConfig
$N = p \cdot q$ \;
\If{$p \cdot q \cdot (S - R_4 - q \cdot iq \cdot (R_3 - R_4)) \not\equiv 0 \mod N r^2$}{
  \Return$\mathsf{error}$
}
\caption{Vigilant's CRT-RSA: Coron \etal fix \#\ref{weak-pq}}
\label{alg-crt-rsa-vigilant-fix-pq}
\end{algorithm}
}
\smallskip
\noindent
What has changed is the recomputation of $p \cdot q$ instead of reusing $N$ in the verification.
The idea is that if $p$ or $q$ is faulted when $N$ is computed then {\sf error} will be returned.
In the original countermeasure, if $p$ or $q$ were to be faulted, the faulted version of $N$ would appear in the condition twice, simplifying itself out and thus hiding the fault.

\section{Formal Methods}
\label{sec-methods}

Our goal is to prove that the proposed countermeasures work,
\ie, that they deliver a result that does not leak information about neither $p$ nor $q$ (if the implementation is subject to fault injection) exploitable in a BellCoRe attack.
In addition, we wish to reach this goal with the two following assumptions:

\begin{itemize}
\item our proof applies to a very general attacker model, and
\item our proof applies to any implementation that is a (strict) refinement of the abstract algorithm.
\end{itemize}

\subsection{CRT-RSA and Fault Injection}
\label{sec-methods-rsa}

First, we must define what computation is done, and what is our threat model.

\begin{definition}[CRT-RSA]
The CRT-RSA computation takes as input a message $M$,
assumed known by the attacker but which we consider to be random,
and a secret key $(p, q, d_p, d_q, i_q)$.
Then, the implementation is free to instantiate any variable, but must return a result equal to:
$S = S_q + q \cdot (i_q \cdot (S_p - S_q) \mod p)$, where:
\begin{itemize}
\item $S_p = M^{d_p} \mod p$, and
\item $S_q = M^{d_q} \mod q$.
\end{itemize}
\label{def-crtrsa}
\end{definition}

\begin{definition}[fault injection]
An attacker is able to request RSA computations, as per Def.~\ref{def-crtrsa}.
During the computation, the attacker can modify any intermediate value by setting it to either a \emph{random value} or \emph{zero}.
At the end of the computation the attacker can read the result.
\label{def-faultinj}
\end{definition}
Of course, the attacker cannot read the intermediate values used during the computation,
since the secret key and potentially the modulus factors are used.
Such ``whitebox'' attack would be too powerful;
nonetheless, it is very hard in practice for an attacker to be able to access reliably to intermediate variables,
due to protections and noise in the side-channel leakage (\eg, power consumption, electromagnetic emanation).
Remark that our model only takes into account fault injection on data;
the control flow is supposed not to be modifiable.

\begin{remark}
We notice that the fault injection model of Def.~\ref{def-faultinj} corresponds to that of Vigilant~\cite{DBLP:conf/ches/Vigilant08}, with the exception that the conditional tests can also be faulted.
To summarize, an attacker can modify a value in the global memory (\emph{permanent fault}), and modify a value in a local register or bus (\emph{transient fault}),
but cannot inject a permanent fault in the input data (message and secret key), nor modify the control flow graph.
\end{remark}

The independence of the proofs on the algorithm implementation demands that the algorithm is described at a high level.
The two properties that characterize the relevant level are as follows:

\begin{enumerate}
\item The description should be low level enough for the attack to work if protections are not implemented.
\item Any additional intermediate variable that would appear during refinement could be the target of an attack,
but such a fault would propagate to an intermediate variable of the high level description, thereby having the same effect.
\end{enumerate}
From those requirements, we deduce that:
\begin{enumerate}
\item The RSA description must exhibit the computation modulo $p$ and $q$ and the CRT recombination;
typically, a completely blackbox description, where the computations would be realized in one go without intermediate variables, is not conceivable.
\item However, it can remain abstract, especially for the computational parts.
For instance a fault in the implementation of the multiplication (or the exponentiation) is either inoffensive, and we do not need to care about it, or it affects the result of the multiplication (or the exponentiation), and our model takes it into account without going into the details of how the multiplication (or exponentiation) is computed.
\end{enumerate}
In our approach, the protections must thus be considered as an augmentation of the unprotected code,
\ie, a derived version of the code where additional variables are used.
The possibility of an attack on the unprotected code attests that the algorithm is described at the adequate level,
while the impossibility of an attack (to be proven) on the protected code shows that added protections are useful in terms of resistance to attacks.
\begin{remark}
The algorithm only exhibits evidence of safety.
If after a fault injection, the algorithm does not simplify to an error detection, then it might only reveal that some simplification is missing.
However, if it does not claim safety, it produces a \emph{simplified} occurrence of a possible weakness to be investigated further.
\end{remark}

\subsection{How \finja Works}
\label{sec-methods-finja}

Our tool\footnote{\url{http://pablo.rauzy.name/sensi/finja.html}}, \finja, works within the framework of modular arithmetic, which is the mathematical framework of CRT-RSA computations.
The general idea is to represent the computation term as a tree which encodes the computation properties.
Our tool then does \emph{symbolic computation} to simplify the term.
This is done by \emph{term rewriting} it, using rules from arithmetic and the properties encoded in the tree.
Fault injections in the computation term are simulated by changing the properties of a subterm, thus impacting the simplification process.
An attack success condition is also given and used on the term resulting from the simplification to check whether the corresponding attack works on it.
For each possible fault injection, if the attack success condition (which may reference the original term as well as the faulted one) can be simplified to $true$ then the attack is said to work with this fault, otherwise the computation is protected against it.
The outputs of our tool are in HTML form: easily readable reports are produced, which contains all the information about the possible fault injections and their outcome.

\subsubsection{Computation Term}
\label{finja-term}

The computation is expressed in a convenient statement-based input language.
This language's Backus Normal Form is given in Fig.~\ref{bnf}.

\begin{figure}
\begin{minipage}[c]{\linewidth}
\begin{minipage}[c]{0.1\linewidth}
~
\end{minipage}
\begin{minipage}[c]{0.85\linewidth}
{\smaller
\begin{Verbatim}[commandchars=@~&,numbers=left]
term    ::= ( stmt )* 'return' mp_expr ';'
stmt    ::= ( decl | assign | verif ) ';'
decl    ::= 'noprop' mp_var ( ',' mp_var )*   @label~bnf:noprop&
          | 'prime' mp_var ( ',' mp_var )*    @label~bnf:prime&
assign  ::= var ':=' mp_expr                  @label~bnf:assign&
verif   ::= 'if' mp_cond 'abort with' mp_expr @label~bnf:verif&
mp_expr ::= '{' expr '}' | expr               @label~bnf:mpe&
expr    ::= '(' mp_expr ')'
          | '0' | '1' | var                   @label~bnf:expr1&
          | '-' mp_expr
          | mp_expr '+' mp_expr
          | mp_expr '-' mp_expr
          | mp_expr '*' mp_expr
          | mp_expr '^' mp_expr
          | mp_expr 'mod' mp_expr             @label~bnf:expr2&
mp_cond ::= '{' cond '}' | cond               @label~bnf:mpc&
cond    ::= '(' mp_cond ')'
          | mp_expr '=' mp_expr               @label~bnf:cond1&
          | mp_expr '!=' mp_expr
          | mp_expr '=[' mp_expr ']' mp_expr  @label~bnf:eqmod&
          | mp_expr '!=[' mp_expr ']' mp_expr @label~bnf:neqmod&
          | mp_cond '/\' mp_cond
          | mp_cond '\/' mp_cond              @label~bnf:cond2&
mp_var  ::= '{' var '}' | var
var     ::= [a-zA-Z][a-zA-Z0-9_']*            @label~bnf:var&
\end{Verbatim}
}
\caption{\label{bnf} BNF of our tool's input language.}
\end{minipage}
\end{minipage}
\end{figure}

A computation term is defined by a list of statements finished by a {\tt return} statement.
Each statement can either:
\begin{itemize}
\item declare a variable with no properties (line~\ref{bnf:noprop});
\item declare a variable which is a prime number (line~\ref{bnf:prime});
\item declare a variable by assigning it a value (line~\ref{bnf:assign}), in this case the properties of the variable are the properties of the assigned expression;
\item perform a verification (line~\ref{bnf:verif}).
\end{itemize}

\noindent
As can be seen in lines~\ref{bnf:expr1} to~\ref{bnf:expr2}, an expression can be:
\begin{itemize}
\item zero, one, or an already declared variable;
\item the sum (or difference) of two expressions;
\item the product of two expressions;
\item the exponentiation of an expression by another;
\item the modulus of an expression by another.
\end{itemize}

\noindent
The condition in a verification can be (lines~\ref{bnf:cond1} to~\ref{bnf:cond2}):
\begin{itemize}
\item the equality or inequality of two expressions;
\item the equivalence or non-equivalence of two expressions modulo another (lines~\ref{bnf:eqmod} and~\ref{bnf:neqmod});
\item the conjunction or disjunction of two conditions.
\end{itemize}

Optionally, variables (when declared using the {\tt prime} or {\tt noprop} keywords), expressions, and conditions can be protected (lines~\ref{bnf:noprop}, \ref{bnf:prime}, \ref{bnf:mpe}, and~\ref{bnf:mpc}, {\tt mp} stands for ``maybe protected'') from fault injection by surrounding them with curly braces.
This is useful for instance when it is necessary to express the properties of a variable which cannot be faulted in the studied attack model.
For example, in CRT-RSA, the definitions of variables $d_p$, $d_q$, and $i_q$ are protected because they are seen as input of the computation.

Finally, line~\ref{bnf:var} gives the regular expression that variable names must match (they start with a letter and then can contain letters, numbers, underscore, and simple quote).

\medskip

After it is read by our tool, the computation expressed in this input language is transformed into a tree (just like the abstract syntax tree in a compiler).
This tree encodes the arithmetical properties of each of the intermediate variable, and thus its dependencies on previous variables.
Properties of intermediate variables can be everything that is expressible in the input language.
For instance, being null or being the product of other terms (and thus, being a multiple of each of them), are possible properties.

\medskip

An example of usage can be found in Appx.~\ref{app-vigilant-orig}, where the original version of Vigilant's countermeasure is written in this language.
Note the evident similarity with the pseudocode of Alg.~\ref{alg-crt-rsa-vigilant}.

\subsubsection{Fault Injection}
\label{finja-fi}

A fault injection on an intermediate variable is represented by changing the properties of the subterm (a node and its whole subtree in the tree representing the computation term) that represents it.
In the case of a fault which forces at zero, then the whole subterm is replaced by a term which only has the property of being null.
In the case of a randomizing fault, by a term which have no properties.

Our tool simulates \emph{all the possible fault injections} of the attack model it is launched with.
The parameters allow to choose:
\begin{itemize}
\item \emph{how many faults} have to be injected (however, the number of tests to be done is impacted by a factorial growth with this parameter, as is the time needed to finish the computation of the proof);
\item \emph{the type} of each fault (\emph{randomizing} or \emph{zeroing});
\item if \emph{transient} faults are possible or if only \emph{permanent} faults should be performed.
\end{itemize}

\subsubsection{Attack Success Condition}
\label{finja-asc}

The attack success condition is expressed using the same condition language as presented in Sec.~\ref{finja-term}.
It can use any variable introduced in the computation term, plus two special variables {\tt \_} and {\tt @} which are respectively bound to the expression returned by the computation term as given by the user and to the expression returned by the computation with the fault injections.
This success condition is checked for each possible faulted computation term.

\subsubsection{Simplification Process}
\label{finja-reduce}

The simplification is implemented as a recursive traversal of the term tree, based on pattern-matching.
It works just like a naive interpreter would, except it does symbolic computation only, by rewriting the term using rules from arithmetic.
Simplifications are carried out in the $\Z$ ring, and its $\Z_N$ subrings.
The tool knows how to deal with most of the $\Z$ ring axioms:
\begin{itemize}
\item the neutral elements ($0$ for sums, $1$ for products);
\item the absorbing element ($0$, for products);
\item inverses and opposites;
\item associativity and commutativity.
\end{itemize}
However, it does not implement distributivity as it is not confluent. Associativity is implemented by flattening as much as possible (``removing'' all unnecessary parentheses), and commutativity is implemented by applying a stable sorting algorithm on the terms of products or sums.

The tool also knows about most of the properties that are particular to $\Z_N$ subrings and applies them when simplifying a term modulo $N$:
\begin{itemize}
\item identity:
  \begin{itemize}
  \item $(a \mod N) \mod N = a \mod N$,
  \item $N^k \mod N = 0$;
  \end{itemize}
\item inverse:
  \begin{itemize}
  \item $(a \mod N) \times (a^{-1} \mod N) \mod N = 1$,
  \item $(a \mod N) + (-a \mod N) \mod N = 0$;
  \end{itemize}
\item associativity and commutativity:
  \begin{itemize}
  \item $(b \mod N) + (a \mod N) \mod N = a + b \mod N$,
  \item $(a \mod N) \times (b \mod N) \mod N = a \times b \mod N$;
  \end{itemize}
\item subrings: $(a \mod N \times m) \mod N = a \mod N$.
\end{itemize}

In addition to those properties a few theorems are implemented to manage more complicated cases where the properties are not enough when conducting symbolic computations:
\begin{itemize}
\item Fermat's little theorem;
\item its generalization, Euler's theorem;
\item Chinese remainder theorem;
\item a particular case of the binomial theorem (see Sec.~\ref{sec-vigilant-orig}).
\end{itemize}

The tool also implements a lemma that is used for modular (in)equality verifications (note that $a$ or $b$ can be $0$):
$a \times x \stackrel{{\smaller ?}}{\equiv} b \times x \mod N \times x$
is simplified to the equivalent test $a \stackrel{{\smaller ?}}{\equiv} b \mod N$.

\section{Analysis of Vigilant's Countermeasure}
\label{sec-analysis}

This section presents and discusses the results of our formal analysis of Vigilant's countermeasure.

\subsection{Original and Repaired Version}

The original version of Vigilant's countermeasure in our tool's input language can be found in Appx.~\ref{app-vigilant-orig}.

In a single fault attack model, we found that the countermeasure has a single point of failure:
if transient fault are possible, then if $p$ or $q$ are randomized in the computation of $N$ at line~\ref{vigi_orig_attack}, then a BellCoRe attack works.

The repaired version by Coron \etal thus does unnecessary modifications, since it fixes three spotted weaknesses (as seen in Sec.~\ref{sec-vigilant-coron}) while there is only one to fix.
Only implementing the modification of Alg.~\ref{alg-crt-rsa-vigilant-fix-pq} is sufficient to be provably protected against the BellCoRe attack.

After that, we went on to see what would happen in a multiple fault model.
What we found is the object of the next section.

\subsection{Our Fixed and Simplified Version}

Attacking with multiple faults when at least one of the fault is a zeroing fault, means that some verification can be skipped by having their condition zeroed.
We found that two of the tests could be skipped without impacting the protection of the computation (the lines~\ref{vigi_useless_test1} and~\ref{vigi_useless_test2}).
Both tests have been removed from our fixed and simplified version, which can be found in Appx.~\ref{app-vigilant-fixed}.
We have proved that removing any other verification makes the computation vulnerable to single fault injection attacks, thereby proving the \emph{minimality} of the obtained countermeasure.

Vigilant's original fault model did not include the possibility to inject fault in the conditions of the verifications.
If we protect them from fault injection, we found that it is still possible to perform a BellCoRe attack exploiting two faults; Woudenberg \etal~\cite{DBLP:conf/fdtc/WoudenbergWM11} showed that two-fault attacks are realistic.
Both faults must be zeroing faults.
One on $R_3$ (resp. $R_4$), and one on $S'_p$ (resp. $S'_q$).
It works because the presence of $R_3$ (resp. $R_4$) in the last verification is not compensated by its presence in $S'_p$ (resp. $S'_q$).
This shows that even security-oriented optimizations cannot be applied mindlessly in security code,
which once again proves the importance of formal methods.

\subsection{Comparison with Aumüller \etal's Countermeasure}

Now that we have (in Appx.~\ref{app-vigilant-fixed}) a version of Vigilant's countermeasure formally proven resistant to BellCoRe attack, it is interesting to compare it with the only other one in this case, namely Aumüller \etal's countermeasure.
Both countermeasures were developed without the help of formal methods, by trial-and-error engineering, accumulating layers of intermediate computations and verifications to patch weaknesses found at the previous iteration.
This was shown, unsurprisingly, to be an ineffective development method.
On one hand, there is no guarantee that all weaknesses are eliminated, and there was actually a point of failure left in the original version of Vigilant's countermeasure.
On the other hand, there is no guarantee that the countermeasure is minimal.
Vigilant's countermeasure could be simplified by formal analysis, we removed two out of its nine verifications.
This proves that ``visual'' code verification as is general practice in many test/certification labs, even for such a concise algorithm is far from easy in practice.
Moreover, the difficulty of this task is increased by the multiplicity of the fault model (permanent vs. transient, randomizing vs. zeroing, single or double).

Apart from these saddening similarities in the development process, there are notable distinctions between the two countermeasures.
Vigilant's method exposes the secret values $p$ and $q$ at two different places, namely during recombination and reduction from modulo $pqr^2$ to $pq$.
Generally speaking, it is thus less safe than Aumüller \etal's method.
In terms of computational complexity the two countermeasures are roughly equivalent.
Aumüller \etal's one does $4$ exponentiations instead of $2$ for Vigilant's, but the $2$ additional ones are with $16$\,bits numbers, which is entirely negligible compared to the prime factors ($p$ and $q$ are $1,024$\,bits numbers).
Vigilant's countermeasure requires more random numbers, which may be costly but is certainly negligible too compared to the exponentiations.

\section{Conclusions and Perspectives}
\label{sec-concl}

We have formally proven the resistance of a fixed version of Vigilant's CRT-RSA countermeasure against the BellCoRe fault injection attack.
Our research allowed us to remove two out of nine verifications that were useless, thereby simplifying the protected computation of CRT-RSA while keeping it formally proved.
Doing so, we believe that we have shown the importance of formal methods in the field of implementation security.
Not only for the development of trustable devices, but also as an \emph{optimization enabler}, both \emph{for speed and security}.
Indeed, Coron \etal's repaired version of Vigilant's countermeasure included fixes for weaknesses that have been introduced by a hasty speed-oriented optimization, (which is never a good idea in the security field),
and two faults attacks are possible only because of random numbers that have been introduced to reinforce the countermeasure.

As a first perspective, we would like to further improve our tool.
It is currently only able to inject faults in the data.
It would be interesting to be able to take into account faults on instructions such as studied by Heydemann \etal~\cite{cryptoeprint:2013:679}.
It would be interesting as well to be able to automatically refine some properties of the variables, for instance to study attacks using a chosen message as input of the algorithm.
Also, multiple fault analyses can take up to several dozens of minutes to compute. Since each attack is independent, our tool would greatly benefit ($2$ to $8$ times speed-up on any modern multi-core computer) from a parallelization of the computations.
Finally, it is worthwhile to note that our tool is {\it ad hoc} in the sense that we directly inject the necessary knowledge into its core.
It would be interesting to see if general purpose tool such as EasyCrypt~\cite{Barthe:2009:POPL} could be a good fit for this kind of work.

\subsubsection*{Acknowledgements}

The authors wish to thank David Vigilant for his insightful answers to our questions.
We also sincerely acknowledge the SIGPLAN PAC program of the ACM for a financial sponsor.

\bibliographystyle{alpha}
\bibliography{sca}

\appendix
\section{Practical Feasibility of the Identified Attacks}
\label{app-practice}

In this section, we discuss the key features that an attacker must control for his attacks to succeed.
We first describe in Sec.~\ref{sub-crt-rsa-implem} the timing characteristics of a CRT-RSA computation implemented on a smartcard.
Second, we analyze in Sec.~\ref{sub-emi_setup} the performance of a fault injection bench.
Third, we analyse in Sec.~\ref{sub-discussion} which fault injection is realistic or, at the contrary, challenging practically speaking.

\subsection{CRT-RSA Algorithm in a Smartcard}
\label{sub-crt-rsa-implem}

We consider a smartcard implementation of CRT-RSA.
It consists in an arithmetic hardware accelerator, capable of computing multiplications and additions.
It is thus termed a MAC (Multiplier-Accumulator Circuit).
This MAC can manipulate data on $64$\,bits.
In $65$\,nm and lower technologies, such a MAC can run at the frequency of $100$\,MHz
(after logical synthesis in a low power standard cell library).
Let us consider that $N$ is a $2048$\,bits number, \ie, $p$ and $q$ hold on $1024$\,bits.
The computations are carried out in the so-called Montgomery representation.
For instance, when computing a modular multiplication, the multiplication and the reductions are interleaved.
Thus, we derive some approximate timings for the unitary operations.
\begin{itemize}
\item One schoolbook addition on $\Z_p$ (or equivalently on $\Z_q$) of two large numbers takes $1024/64=16$ clock cycles, \ie, $160$\,ns;
\item One schoolbook squaring or multiplication on $\Z_p$ of two large numbers takes $2 \times (1024/64)^2=512$ clock cycles, \ie, $5.12$\,$\mu$s;
The factor two takes into the account the reduction steps.
\item One schoolbook exponentiation on $\Z_p$ of one large number by another large number takes $512 \times (1024 \times \nicefrac32)=786432$ clock cycles, \ie, about $7.8$\,ms;
This estimation assumes the square-and-multiply exponentiation algorithm, in which there are as many squarings as bits in $p$,
but only half (in average) multiplications.
\end{itemize}
Besides, we notice that some operations in Alg.~\ref{alg-crt-rsa-vigilant} are hybrid:
they involve one large number and one small number (\eg, $64$\,bits number, of the size of $r$, $R_1$, $R_2$, $R_3$ or $R_4$).
Obviously, those operations unfold more quickly than operations with two large numbers.
Some operations in Alg.~\ref{alg-crt-rsa-vigilant-fix-p} can even be implemented to be extremely fast.
For instance, the subtraction on line~\ref{pminusone_1} or \ref{pminusone_2} of Alg.~\ref{alg-crt-rsa-vigilant-fix-p} (resp. Alg.~\ref{alg-crt-rsa-vigilant-fix-q}) can be imagined to be done in one clock cycle, \ie, $10$\,ns.
Indeed, as $p$ and $q$ are odd (because they are large prime numbers),
this ``decrementation'' consists simply in resetting the least significant bit of the operand $p$ (resp. $q$).

So, concluding, the operations in the CRT-RSA algorithm can be of diverse durations,
ranging $6$ orders of magnitude,
from $10$\,ns for operations computable in one clock cycle,
to about $8$\,ms for one modular exponentiation.

\subsection{Fault Injection Setup}
\label{sub-emi_setup}

An efficient fault injection method consists in the production of a transient electromagnetic (EM) pulse of high energy in the vicinity of the circuit that runs the CRT-RSA.
This kind of fault is termed EMI (EM Injection).
Some setups can feature the possibility of multiple fault injections.
Simple, double and triple EM pulses are illustrated respectively in Fig.~\ref{img_81160A_pulse_png}, \ref{img_81160A_pulse2O_png} and \ref{img_81160A_pulse3O_png}.

\begin{figure}[h!]
\center
\includegraphics[width=0.65\linewidth] % ePrint
  {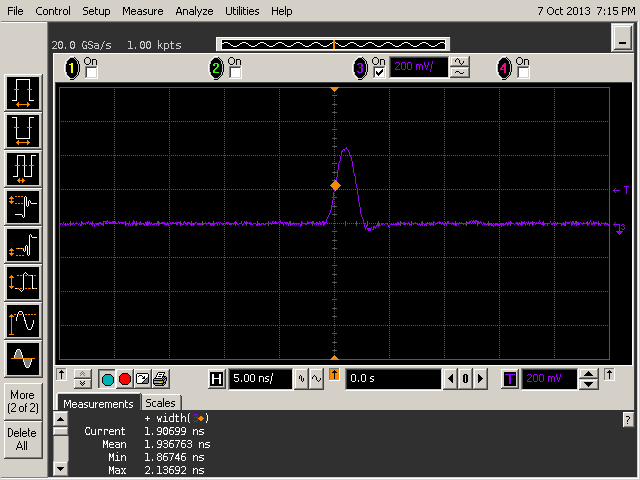}
\caption{First-order fault injection attack --- proof of concept}
\label{img_81160A_pulse_png}
\end{figure}
\begin{figure}[h!]
\center
\includegraphics[width=0.65\linewidth] % ePrint
  {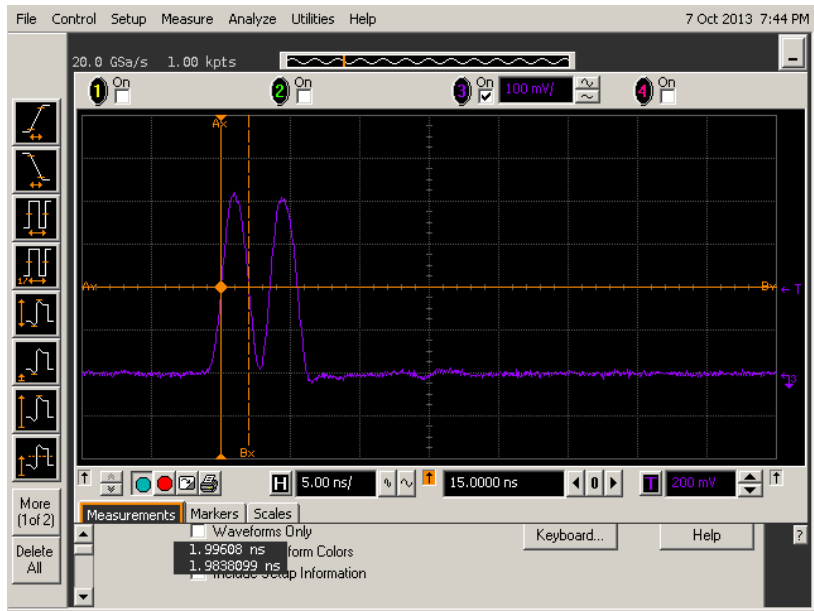}
\caption{Second-order fault injection attack --- proof of concept}
\label{img_81160A_pulse2O_png}
\end{figure}
\begin{figure}[h!]
\center
\includegraphics[width=0.65\linewidth] % ePrint
  {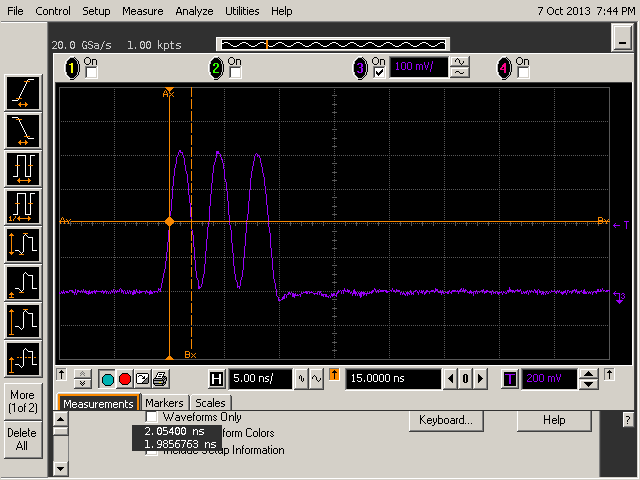}
\caption{Third-order fault injection attack --- proof of concept}
\label{img_81160A_pulse3O_png}
\end{figure}

We detail below the technical characteristics of a fairly high-end bench%
\footnote{We credit Laurent Sauvage for this information,
that comes from the specification of one plug-in of the Smart-SIC Analyzer tool of Secure-IC S.A.S.}:
\begin{itemize}
\item Pulse amplitude: \dotfill $0$---$500$\,mV.
\item Pulse duration:  \dotfill $2$\,ns.
\item Repeatability:   \dotfill $2$\,ns $\iff$ $500$\,MHz.
\end{itemize}

\subsection{Practicality of the Attacks Identified by \textsf{finja}}
\label{sub-discussion}

The \emph{pulse amplitude} (or ``\emph{energy}'') shall be sufficiently high for the EMI to have an effect, and not too large to avoid chip damage.
A typical method to set the adequate level of energy is to run the CRT-RSA computations iteratively, while increasing the energy from zero, with small increments (\eg, a few millivolts).
After a given number of iterations, the CRT-RSA signature will be faulty or the circuit under attack will refuse to give an answer.
From this value on, the EMI starts to have an effect on the chip.
Perhaps this fault injection is probabilistic, meaning that it does not always have an effect.
So, it is safe to increase the energy level of a few percents to raise the injection probability of success as close to one as possible.

The \emph{pulse duration} is an important parameter, because it conditions the accuracy in time of the stress.
As mentioned in Sec.~\ref{sub-crt-rsa-implem}, there is a great variability in the duration of the operations.
Therefore, some attacks will be more difficult to set up successfully.
For exemple, the two corrections (Alg.~\ref{alg-crt-rsa-vigilant-fix-p} and~\ref{alg-crt-rsa-vigilant-fix-pq}) finally repair a flaw that is fairly difficult to fault,
because the operations last:
\begin{itemize}
\item $1$~clock cycle, \ie, $10$\,ns, and
\item $512$~clock cycles, \ie, $5.12$\,$\mu$s respectively.
\end{itemize}

Implementations of CRT-RSA are generally working in constant time, to prevent timing attacks~\cite{kocher-timing_attacks}.
Those attacks are indeed to be taken very seriously, as they can even be launched without a physical access to the device.
Paradoxically, this helps carry out timely attacks on focused stages in the algorithm.
So, with the setup described in Sec.~\ref{sub-emi_setup}, such fault attacks remain possible.
Also, multiple faults can be injected easily if the algorithm timing is deterministic,
because the \emph{repetition rate} of the setup is greater than once per clock cycle.
Indeed, a new pulse can be generated every other period of $2$\,ns, \ie, at a maximal rate of $500$\,MHz.

However, a bench of lower quality would deny some exploits, for instance those that require:
\begin{inparaenum}[(\itshape a\upshape)]
\item to fault an operation of small duration;
\item for high-order fault attacks, operations that are too close one from each other.
\end{inparaenum}

% \clearpage % ACM

\begin{minipage}[c]{\textwidth} % ePrint
\begin{minipage}[c]{0.42\textwidth} % ePrint

\section{Vigilant's Original Countermeasure}
\label{app-vigilant-orig}

\begin{minipage}[c]{\linewidth}
\begin{minipage}[c]{0.1\linewidth}
~
\end{minipage}
\begin{minipage}[c]{0.85\linewidth}
{\smaller
\begin{Verbatim}[commandchars=&~|,numbers=left]
noprop error, M, e, r, R1, R2, R3, R4 ;
prime {p}, {q} ;

dp := { e^-1 mod (p-1) } ;
dq := { e^-1 mod (q-1) } ;
iq := { q^-1 mod p } ;

~&color~gray|--- p part ---|
p' := p * r * r ;
Mp := M mod p' ;
ipr := p^-1 mod (r * r) ;
Bp := p * ipr ;
Ap := 1 - Bp mod p' ;
M'p := Ap * Mp + Bp * (1 + r) mod p' ;

if M'p !=[p] M abort with error ;

d'p := dp + R1 * (p - 1) ;
Spr := M'p^d'p mod p' ;

if d'p !=[p - 1] dp abort with error ;

if Bp * Spr !=[p'] Bp * (1 + d'p * r)
  abort with error ;

S'p := Spr - Bp * (1 + d'p * r - R3) ;

~&color~gray|--- q part ---|
q' := q * r * r ;
Mq := M mod q' ;
iqr := q^-1 mod (r * r) ;
Bq := q * iqr ;
Aq := 1 - Bq mod q' ;
M'q := Aq * Mq + Bq * (1 + r) mod q' ;

if M'q !=[q] M abort with error ;

if Mp !=[r * r] Mq abort with error ; &label~vigi_useless_test1|

d'q := dq + R2 * (q - 1) ;
Sqr := M'q^d'q mod q' ;

if d'q !=[q - 1] dq abort with error ;

if Bq * Sqr !=[q'] Bq * (1 + d'q * r)
  abort with error ;

S'q := Sqr - Bq * (1 + d'q * r - R4) ;

~&color~gray|--- recombination ---|
S := S'q + q * (iq * (S'p - S'q) mod p') ;
N := p * q ; &label~vigi_orig_attack|

if N * (S - R4 - q * (iq * (R3 - R4)))
  !=[N * r * r] 0 abort with error ;

if q * iq !=[p] 1 abort with error ;  &label~vigi_useless_test2|

return S mod N ;

%%

_ != @ /\ ( _ =[p] @ \/ _ =[q] @ )
\end{Verbatim}
}
\end{minipage}
\end{minipage}

\end{minipage} % ePrint
\begin{minipage}[c]{0.05\textwidth} % ePrint
  ~ % ePrint
\end{minipage} % ePrint
\begin{minipage}[c]{0.42\textwidth} % ePrint

\section{Fixed Vigilant's Countermeasure}
\label{app-vigilant-fixed}

\begin{minipage}[c]{\linewidth}
\begin{minipage}[c]{0.1\linewidth}
~
\end{minipage}
\begin{minipage}[c]{0.85\linewidth}
{\smaller
\begin{Verbatim}[commandchars=&~|,numbers=left]
noprop error, M, e, r, R1, R2, R3, R4 ;
prime {p}, {q} ;

dp := { e^-1 mod (p-1) } ;
dq := { e^-1 mod (q-1) } ;
iq := { q^-1 mod p } ;

~&color~gray|--- p part ---|
p' := p * r * r ;
Mp := M mod p' ;
ipr := p^-1 mod (r * r) ;
Bp := p * ipr ;
Ap := 1 - Bp mod p' ;
M'p := Ap * Mp + Bp * (1 + r) mod p' ;

if M'p !=[p] M abort with error ;

d'p := dp + R1 * (p - 1) ;
Spr := M'p^d'p mod p' ;

if d'p !=[p - 1] dp abort with error ;

if Bp * Spr !=[p'] Bp * (1 + d'p * r)
  abort with error ;

S'p := Spr - Bp * (1 + d'p * r - R3) ;

~&color~gray|--- q part ---|
q' := q * r * r ;
Mq := M mod q' ;
iqr := q^-1 mod (r * r) ;
Bq := q * iqr ;
Aq := 1 - Bq mod q' ;
M'q := Aq * Mq + Bq * (1 + r) mod q' ;

if M'q !=[q] M abort with error ;

~&color~gray|-- useless verification removed|

d'q := dq + R2 * (q - 1) ;
Sqr := M'q^d'q mod q' ;

if d'q !=[q - 1] dq abort with error ;

if Bq * Sqr !=[q'] Bq * (1 + d'q * r)
  abort with error ;

S'q := Sqr - Bq * (1 + d'q * r - R4) ;

~&color~gray|--- recombination ---|
S := S'q + q * (iq * (S'p - S'q) mod p') ;
N := p * q ;

if ~&color~blue|p * q| * (S - R4 - q * (iq * (R3 - R4)))
  !=[N * r * r] 0 abort with error ;

~&color~gray|-- useless verification removed|

return S mod N ;

%%

_ != @ /\ ( _ =[p] @ \/ _ =[q] @ )
\end{Verbatim}
}
\end{minipage}
\end{minipage}

\end{minipage} % ePrint
\end{minipage} % ePrint

\end{document}